\documentclass[12pt]{amsart}

\usepackage{amsmath,amsthm,amssymb,mathbbol,amsfonts,verbatim,algorithm}
\usepackage[colorlinks,
            linkcolor=blue,
            anchorcolor=blue,
            citecolor=blue
            ]{hyperref}
\usepackage[hmargin=1.2in,vmargin=1.2in]{geometry}
\usepackage{array}
\usepackage{multirow}
\usepackage{algorithmic}

\title[Nonlinear codes]{Nonlinear codes with low redundancy}

\author{Shu Liu}\address{National Key Laboratory on Wireless Communications, University of Electronic Science and Technology of China, Chengdu, China} \email{shuliu@uestc.edu.cn}
\author{Chaoping Xing} \address{School of Electronic Information and Electric Engineering, Shanghai Jiao Tong University, Shanghai, China} \email{xingcp@sjtu.edu.cn}

\date{}

\newtheorem{lemma}{Lemma}[section]
\newtheorem{theorem}[lemma]{Theorem}
\newtheorem{cor}[lemma]{Corollary}

\newtheorem{ex}[lemma]{Example}

\theoremstyle{remark}
\newtheorem{rmk}{Remark}

\renewcommand{\epsilon}{\varepsilon}
\renewcommand{\le}{\leqslant}
\renewcommand{\ge}{\geqslant}

\newcommand{\vnote}[1]{}

\def\ZZ{\mathbb{Z}}

\def\F{\mathbb{F}}
\def \mC {\mathcal{C}}
\def \mA {\mathcal{A}}

\def \mA {\mathcal{A}}
\def \mB {\mathcal{B}}
\def \mC {\mathcal{C}}

\def \Xi {{X^{[i]}}}

\newcommand{\Ga}{\alpha}

\def \bGa{{\boldsymbol\Ga}}

\def \bb {{\bf b}}
\def \bc {{\bf c}}

\def \bx {{\bf x}}

\def \by {{\bf y}}
\def \bs {{\bf s}}
\def \bz {{\bf z}}
\def \bu {{\bf u}}
\def \bv {{\bf v}}

\def\bh{{\bf h}}

\numberwithin{equation}{subsection}

\def\GRS{{\rm GRS}}
\begin{document}
\maketitle

\begin{abstract}
Determining the largest size, or equivalently finding the lowest redundancy, of $q$-ary codes for given length  and minimum distance  is one of the central and fundamental problems in coding theory.  Inspired by the construction of Varshamov-Tenengolts (VT for short) codes via check-sums, we provide an explicit construction of nonlinear codes with lower redundancy than linear codes under the same length  and minimum distance. Similar to the VT codes, our construction works well for small distance (or even constant distance). Furthermore, we design quasi-linear time decoding algorithms for both erasure and adversary errors.

\end{abstract}

\section{Introduction}
Given a code alphabet, determining the largest size $A_q(n,d)$  of $q$-ary block codes $\mC$ with length $n$ and Hamming minimum distance $d$ has been a long-standing problem in coding theory.
There are a large amount of papers in literatures to study the quantity $A_q(n,d)$. When $d$ is small or even a constant, BCH codes and Hamming codes usually have large size in this parameter regime. The other way to see whether a code has good parameters, one can simply look into its redundancy.  Thus, we can define the smallest possible redundancy by $r_q(n,d)=n-\log_qA_q(n,d)$ for given alphabet size $q$, length $n$ and minimum distance $d$. Of course, one would like to have this redundancy  $r_q(n,d)$ as small as possible.
In this paper, we mainly focus on $r_q(n,d)$ (or equivalently $A_q(n,d)$) with small $d$.

\subsection{Known results}
There are various upper bounds on $A_q(n,d)$ (see \cite{LX2004}) such as the Hamming bound (a relatively tight bound  for small distance like $d=3$), the Plotkin bound (a relatively tight bound  for large distance such as $d> (q-1)n/q$), the Grismer bound (a  bound  only for linear codes) and  the linear programming bound. A benchmark for a good code is the Gilbert-Varshamov bound--a lower bound on $A_q(n,d)$.  There are two versions of the Gilbert-Varshamov bound, one is called the weak version and the other is called strong version (see Section 2 for the detail). The strong version is applicable only for linear codes (see \cite{LX2004}).

Apart from some well-known families of codes such as Hamming codes, Reed-Solomon codes, BCH codes, Reed-Muller codes, Goppa codes and algebraic geometry codes, etc, there are also various constructions of linear and nonlinear codes  that provide lower bounds on $A_q(n,d)$ for some specific parameters. As there are too many such constructions in literatures, we are not going to mention these constructions one by one. {The reader may refer some books \cite{LX2004,MS1983,M2005,V1992} and the online table \cite{online1} for some of these constructions.} We would like to particularly discuss Hamming codes. Hamming codes have minimum distance $3$ and achieve the Hamming bound. Therefore,  Hamming codes are optimal in the sense that the codes achieve the maximal size $A_q(n,3)$. However, there are  some constraints on parameters for Hamming codes. Firstly, as Hamming codes are linear codes, thus we require that $q$ is a prime power. Secondly, the code length of Hamming codes are of the form $\frac{q^r-1}{q-1}$ for $r\ge 2$. Thus, except for some specific parameters, in general we do not know exact values of $A_q(n,3)$ if $q$ is not a prime power; or $n$ is not of the form $\frac{q^r-1}{q-1}$. We do not even know if $r_q(n,3)<n-k_q(n,3)$ for a prime power $q$, where $k_q(n,d)$ is the maximum dimension $k$ such that there exists a $q$-ary $[n,k,3]$-linear code (this means that $n-k_q(n,3)$ is the smallest redundancy for $q$-ary linear codes of length $n$ and distance $3$).

\subsection{Our results}
In this paper, we present a construction of nonlinear codes with low redundancy that is inspired by the construction of Varshamov-Tenengolts codes via {check-sums}. As a result, in general our codes have lower redundancy than linear codes for given code length and distance. In particular, when minimum distance $d$ is fixed,  one can show that our codes have smaller redundancy than linear codes if the code length $n$ lies in some intervals.

 Although it is generally difficult to design efficient algorithms for nonlinear codes, we present quasi-linear time decoding algorithms for our codes for both erasure and adversary errors. More precisely speaking, the decoding complexity is $O(n\log^4n)$ bit operations for both erasure and adversary errors. In addition, if distance $d$ is constant, the  decoding complexity is $O(\log^2n)$ and $O(n\log^2n)$ bit operations for erasure and adversary errors, respectively.

\subsection{Organization of the paper}
The paper is organized as follows. In Section 2, we present some preliminaries on codes including definitions of $A_q(n,d)$, $r_q(n,d)$, the main MDS conjecture and the relation between distance and erasure error correcting capability. In Section 3, we provide an explicit construction of the nonlinear code with low redundancy. Furthermore, some numerical examples are given in Section 3 as well. In the last section, decoding algorithms of our nonlinear codes constructed in Section 3 are presented.

\section{Preliminary}
\subsection{Some notations on codes}
 Let $\F_q$ be a finite field with $q$ elements.
For two integers $a, b$ with $a\le b$, denote by $[a,b]$ the set $\{a,a+1,\dots,b\}$. In particular, denote by $[n]$ the set $[1,n]$. A $q$-ary code $\mC$ of length $n$ is a subset of $[0,q-1]^n$. If the size of $\mC$ is $M$, we say that $\mC$ is a $q$-ary $(n,M)$-code or an $(n,M)_q$-code. Furthermore, if the Hamming distance of $\mC$ is $d$, we say that $\mC$ is  a $q$-ary $(n,M,d)$-code. It is well known that a  code with minimum Hamming distance  $d$ can correct $\lfloor\frac{d-1}{2}\rfloor$ adversary errors.
It is clear that if there is a $q$-ary $(n,M,d_1)$-code with $d_1>d$, then there is also a $q$-ary $(n,M,d)$-code. This is because we can turn a $q$-ary $(n,M,d_1)$-code into a $q$-ary $(n,M,d)$-code by changing every codeword of a fixed set of $d_1-d$ positions to $0$.
For a $q$-ary linear code $\mC$ with length $n,$ dimension $k$ and minimum distance $d$, we say that $\mC$ is a $q$-ary $[n,k,d]$-linear code. 
\subsection{Optimal linear and nonlinear codes}
In coding theory, it is a great challenge to determine the maximal size of $q$-ary codes for given length $n$ and minimum distance $d$. The following quantity characterizes this maximum size:
\begin{equation}\label{eq:1}
A_q(n,d)=\max\{M\in\ZZ_{>0}:\; \mbox{there is a $q$-ary $(n,M,d)$ code}\}.
\end{equation}
An $(n,M,d)_q$-code is called optimal if $M=A_q(n,d)$.

Similarly, when $q$ is a prime power, we can define maximal size of $q$-ary linear codes for given length $n$ and minimum distance $d$ as follows:
\begin{equation}\label{eq:2}
K_q(n,d)=\max\{M\in\ZZ_{>0}:\; \mbox{there is a $q$-ary $(n,M,d)$-linear code}\}.
\end{equation}
An $[n,k,d]_q$-code is called optimal if $k=\log_qK_q(n,d)$.

By the simple prorogation rules mentioned above (also see \cite[Chapter 6]{LX2004}), we know that $A_q(n,d)$ is a decreasing function of $d$ for given $q$ and $n$, while $A_q(n,d)$ is an increasing function of $n$ for given $q$ and $d$. Usually the quantities $A_q(n,d)$ and $K_q(n,d)$ are large integers. Thus, we define the following two normalized quantities
\begin{equation}\label{eq:3}
a_q(n,d)=\log_qA_q(n,d);\quad k_q(n,d)=\log_qK_q(n,d).
\end{equation}
Note that $a_q(n,d)$ may not be an integer, while $k_q(n,d)$ is always an integer which is the dimension of a code.

Recall, $r_q(n,d)=n-a_q(n,d)$ is the smallest redundancy of $q$-ary codes of length $n$ and minimum distance $d$. Now we define $r^L_q(n,d)=n-k_q(n,d)$ to be the smallest redundancy of $q$-ary linear codes of length $n$ and minimum distance $d$. Then it is clear that $r_q(n,d)\le r_q^L(n,d)$. It is not clear where $r_q(n,d)$ is strictly less than $r_q^L(n,d)$ in general.

\subsection{Defect} In this subsection, we assume that $q$ is a prime power. Then by the Singleton bound, we know that every $q$-ary $[n,k,d]$-linear code $\mC$ obeys
\begin{equation}\label{eq:x1}
k\le n-d+1.
\end{equation}
When the equality in \eqref{eq:x1} holds, $\mC$ is called a maximum distance separable (MDS for short) code.

One of the main problems for MDS codes is to determine the maximum length of an MDS code. The following is an important conjecture.

{\textbf{ Main Conjecture on MDS codes.}} For a nontrivial $q$-ary $[n,k,n-k+1]$-MDS code, we have
\[n\le \left\{\begin{array}{ll}
q+2&\mbox{ if $2|q$ and $k\in\{3,q-1\}$};\\
q+1&\mbox{ otherwise}.
\end{array}
\right.\]

A $q$-ary $[n,k,d]$-linear code satisfying
\begin{equation}\label{eq:x2}
k= n-d+1-r
\end{equation}
is said to have defect $r$. It is clear that codes with defect $0$ are MDS codes. A code with defect $1$ is called a almost MDS code. We denote by $N_q(d,r)$ the largest length $n$ of  $q$-ary $[n,n-d+1-r,d]$-linear codes.


By generalized Reed-Solomon codes and their extended codes, we know that $N_q(d,0)\ge q+1$. On the other hand, the main MDS conjecture tells us  that $N_q(d,0)\le q+2$.

The values $N_q(3,0)$ and $N_q(3,1)$ are completely determined \cite{BB1952}. For completeness, we provide a short proof below.
\begin{lemma}\label{lem:2.1} For $r\ge 0$, one has
 \[N_q(3,0)= q+1,\quad N_q(3,1)=q^2+q+1.\]
\end{lemma}
\begin{proof} If $\mC$ is an $[n,n-2,3]$-MDS code, then its parity-check matrix has size $2\times n$. As the distance of $\mC$ is $3$, any two coloumns are linearly independent. This means that one-dimensional spaces spanned by columns of $H$ are pairwise distinct. As there are $\frac{q^2-1}{q-1}=q+1$ one-dimensional spaces in $\F_q^2$, the desired result follows.

The similar arguments can be used to show $N_q(3,1)=\frac{q^3-1}{q-1}=q^2+q+1$. In this case, we consider one-dimensional spaces in $\F_q^3$.
\end{proof}

\subsection{Erasure errors and minimum distance}
It is well known that minimum distance of a code determines erasure error correcting capability. Informally, we say that a $q$-ary code $C\subseteq [0,q-1]^n$ can correct $\tau$  erasure errors if any $\tau$ positions of a codeword are erased, we can still recover this codeword. Precisely speaking, a $q$-ary code $\mC\subseteq [0,q-1]^n$ can correct $\tau$  erasure errors if for any subset $S\subset[n]$ with $|S|= n-\tau$ and a codeword $\bc\in\mC$, no other codewords $\bb\in\mC$ satisfy $\bb_S=\bc_S$, where $\bb_S$ is the projection of $\bb$ at $S$. The following lemma follows immediately.
\begin{lemma}\label{lem:2.2}
A code $\mC$ has minimum distance at least $d$ if and only if it can correct $d-1$ erasure errors.
\end{lemma}

\subsection{Generalized Reed-Solomon codes}
Let $F$ be a field and choose nonzero elements $\{v_1,v_2,\dots,v_n\}\subseteq F^*$ ($v_i$ are not necessarily distinct) and pairwise distinct elements $\{\Ga_1,\Ga_2,\dots,\Ga_n\}\subseteq F$. Put  $\bv=(v_1,\cdots, v_n)\in (F^*)^n$ and  $\bGa=(\Ga_1,\cdots, \Ga_n)\in F^n.$ For $0\le k\le n,$ the generalized Reed-Solomon  is defined by
\[\GRS_{n,k}(\bGa,\bv)=\{v_1f(\Ga_1),\cdots, v_n f(\Ga_n): f\in F[x]_{<k}\}.\]
Here, $F[x]_{<k}$ denotes the set of polynomial in $F[x]$ of degree less than $k$.

\begin{lemma}
$\GRS_{n,k}(\bGa,\bv)$ is an $[n,k,d]$-linear code over $F$ with length $n\le |F|.$ If $0<k< n,$ then $d=n-k+1.$ In particular, a generalized Reed-Solomon code is an MDS code.
\end{lemma}

The dual of the  generalized Reed-Solomon code $\GRS_{n,k}(\bGa,\bv)$ is given by
\[\GRS_{n,k}(\bGa,\bv)^\perp=\GRS_{n,n-k}(\bGa,\bu),\]
where $\bu=(u_1,\cdots, u_n)$ with $u_i^{-1}=v_i\Pi_{j\neq i}(\Ga_i-\Ga_j).$ Thus, the dual code of a generalized Reed-Solomon code is also an MDS code.

Decoding of generalized Reed-Solomon codes is  of both practical and theoretical importantance. The most widely known decoding is the syndrome-based Reed-Solomon codes decoding, in which the key equation is solved using either the Berlekamp-Massey algorithm, the Euclidean algorithm or fast Fourier tranform. For an $[n,k]$-generalized Reed-Solomon code, the computational complexity of syndrome-based decoding is $O(n\log n+(n-k)\log^2(n-k))$ operations of field elements \cite{{TH2002}}.
\section{Construction}

In this section, we provide an explicit construction of nonlinear code with low redundancy. Our construction follows the idea of VT codes with additional parity-check from {\color{red}a} RS code.  Some numerical examples are provided to show that our codes have lower redundancy than linear codes for  given length $n$ and minimum distance $d$.

Let $q$ be an integer greater than $1$ and let $d,n$ be two positive integers greater than $2$. Let  $\ell$ be the smallest prime satisfying $\ell\ge \max\{n,q\}$.  For $\bx=(x_1,x_2,\dots,x_n)\in[0,q-1]^n$, define the functions
\begin{equation}\label{eq:5}
\quad t_j(\bx)=\sum_{i=1}^ni^jx_i
\end{equation}
for $j\ge 0$.

Choose  $b_0,b_1,b_1,\dots,b_{d-2}$ with $b_0\in[0,(d-1)(q-1)]$ and $b_i\in[0,\ell-1]$ for $1\le i\le d-2$ and put $\bb=(b_0,b_1,\dots,b_{d-2})$.
Define the code
\begin{equation}\label{eq:6} \begin{aligned}
\mC_d(\bb)=\{\bGa:=(\Ga_1,\Ga_2,\dots,\Ga_n)\in[0,q-1]^n:\; t_j(\bGa)\equiv b_j\pmod{\ell}\\ \mbox{  for $1\le j\le d-2$};\;
 t_0(\bGa)\equiv b_0\pmod{(d-1)(q-1)+1} \}
  \end{aligned}
 \end{equation}

 Now we show that the minimum distance of  the code $\mC_d(\bb)$  is at least $d$.
 \begin{lemma}\label{lem:3.1} For $d\ge 3$, the code $\mC_d(\bb)$ given in \eqref{eq:6} has Hamming distance at least $d$.
 \end{lemma}

\begin{proof} By Lemma \ref{lem:2.2}, it suffices to show that the code  $\mC_d(\bb)$ can correct $d-1$ erasure errors. Let $S:=\{k_1,k_2,\dots,k_{d-1}\}$ be a subset of $[n]$ with $k_1<k_2<\cdots<k_{d-1}$. Assume that a codeword $\bGa\in \mC_d(\bb)$ are erased at positions of $S$, i.e., $(\Ga_{k_1},\Ga_{k_2},\dots,\Ga_{k_{d-1}})$ is erased. Put $c=\sum_{i\in[n]\setminus S}\Ga_i$. Then we can compute \[\sum_{i\in S}\Ga_i=\sum_{i=1}^n\Ga_i- \sum_{i\in[n]\setminus S}\Ga_i\equiv b_0-c\pmod{(d-1)(q-1)+1},\] since $t_0(\bGa)=b_0$. Let $c_0\in[0,(d-1)(q-1)]$ with $c_0\equiv b_0-c\pmod{(d-1)(q-1)+1}$. Then we have $\sum_{i\in S}\Ga_i=c_0$ since $\sum_{i\in S}\Ga_i\le |S|(q-1)=(d-1)(q-1)$.
Hence, we have $\sum_{i\in S}\Ga_i\pmod{\ell}=c_0$.

For $1\le j\le d-2$, we can also compute
\[c_j:=b_j-\sum_{i\in[n]\setminus S}i^j\Ga_i\pmod{\ell}= \sum_{i=1}^ni^j\Ga_i-\sum_{i\in[n]\setminus S}i^j\Ga_i\pmod{\ell}=\sum_{i\in S}i^j\Ga_i\pmod{\ell}.\]
This means that the vector $(\Ga_{k_1},\Ga_{k_2},\dots,\Ga_{k_{d-1}})$ is the unique solution of the following equation
\[\begin{pmatrix}
1&1&\cdots&1\\
k_1&k_2&\cdots&k_{d-1}\\
\vdots&\vdots&\vdots&\vdots\\
k_1^{d-2}&k_2^{d-2}&\cdots&k_{d-1}^{d-2}
\end{pmatrix}
\bx^T=\bc^T\pmod{\ell},
\]
where $\bc=(c_0,c_1,\dots,c_{d-2})$. Note that the above matrix is a $(d-1)\times (d-1)$ Vandermonde matrix which is invertible.
This completes the proof.
\end{proof}
The following lower bound on $A_q(n,d)$ can be easily derived from the above Lemma.
 \begin{theorem}\label{thm:3.2} For $d\ge 3$, let $\ell$ be the smallest prime satisfying $\ell\ge\max\{q,n\}$. Then
 one has
 \[A_q(n,d)\ge \frac{q^n}{((d-1)(q-1)+1) \ell^{d-2}}.\]

 \end{theorem}
\begin{proof} By Lemma \ref{lem:3.1}, it suffices to show that there exists a vector $\bb\in[0,(d-1)(q-1)]\times [0,\ell-1]^{d-2}$ such that the code $\mC_d(\bb)$ is an $(n,M)$ code with
\[M\ge \frac{q^n}{((d-1)(q-1)+1) \ell^{d-2}}.\]

It is clear that
\[\bigcup_{\bb\in[0,(d-1)(q-1)]\times[0,\ell-1]^{d-2}}\mC_d(\bb)=[0,q-1]^n.\]
This gives
\[q^n=\left|[0,q-1]^n\right|=\left|\bigcup_{\bb\in[0,(d-1)(q-1)]\times[0,\ell-1]^{d-2}}\mC_d(\bb)\right|\le\sum_{\bb\in[0,(d-1)(q-1)]\times[0,\ell-1]^{d-2}}|\mC_d(\bb)|.\]
This implies that when  $\bb$ runs through  $[0,(d-1)(q-1)]\times[0,\ell-1]^{d-2}$, the average size of $\mC_d(\bb)$ is at least $ \frac{q^n}{((d-1)(q-1)+1) \ell^{d-2}}$. The proof is completed.
\end{proof}

\begin{ex}{\rm In this example, we show that $r_q(n,3)$ is strictly less than $r_q^L(n,3)$ for some parameters.

\begin{itemize}
\item[(i)]
Take $q=4$ and $d=3$.
\begin{itemize}
\item[1)]For $n\in[30,31]$, we can take $\ell=31$, then we have $\ell\ge\max\{q,n\}$. By Theorem \ref{thm:3.2}, we have $r_4(n,3)=n-a_4(n,3)\le 3.88.$ On the other hand, by the online table of \cite{online}, we have $r_4^L(n,3)=4$ for $30\le n\le 31$.


 \item[2)]   For $n\in[24,29]$, we can take $\ell=29$, then we have $\ell\ge\max\{q,n\}$. By Theorem \ref{thm:3.2}, we have  $r_4(n,3)=n-a_4(n,3)\le 3.83.$ On the other hand, by the online table of \cite{online}, we have $k_4(n,3)=n-4$ which implies $r_4^L(n,3)=4$ for $24\le n\le 29$.


\item[3)]    Similarly, for $n\in[22,23]$, we can take $\ell=23$, then we have $\ell\ge\max\{q,n\}$. By Theorem \ref{thm:3.2}, we have $ r_4(n,3)=n-a_4(n,3)\le 3.67.$
    On the other hand, by the online table of \cite{online}, we have $r_4^L(n,3)=4$ for $22\le n\le 23$.

  \item[4)]   For $n\in[86,87]$, we can take $\ell=87$, then we have $\ell\ge\max\{q,n\}$. By Theorem \ref{thm:3.2}, we have $r_4(n,3)=n-a_4(n,3)\le 4.63.$ 
     On the other hand, by the online table of \cite{online}, we have $r_4^L(n,3)=5$ for $86\le n\le87$.
     \end{itemize}

     We have many other instances of parameters for which our Theorem \ref{thm:3.2} shows that $r_4(n,3)<r^L_4(n,3)$. We tablet some of these parameters in the following table.
     \begin{center}
      Table I\\{Nonlinear codes with lower redundancy when $q=4$ and $d=3$}
      \medskip
     \begin{tabular}{|c|c|c|c|c|} \hline\hline
     $q$&$n$&$d$&Upper bound on $r_q(n,d)$& $r^L_q(n,d)$\\ \hline
       4&22&3&3.67&4\\ \hline
        4&23&3&3.67&4\\ \hline
       4&24&3&3.83&4\\ \hline
        4&25&3&3.83&4\\ \hline
         4&26&3&3.83&4\\ \hline
         4&27&3&3.83&4\\ \hline
        4&28&3&3.83&4\\ \hline
         4&29&3&3.83&4\\ \hline
         4&30&3&3.88&4\\ \hline
        4&31&3&3.88&4\\ \hline
          4&86&3&4.63&5\\ \hline
       4&87&3&4.63&5\\ \hline
        4&88&3&4.64&5\\ \hline
       4&89&3&4.64&5\\ \hline
        4&90&3&4.7&5\\ \hline
         4&91&3&4.7&5\\ \hline
         4&92&3&4.7&5\\ \hline
        4&93&3&4.7&5\\ \hline
         4&94&3&4.7&5\\ \hline
         4&95&3&4.7&5\\ \hline
        4&96&3&4.7&5\\ \hline
          4&97&3&4.7&5\\ \hline\hline
     \end{tabular}
     \end{center}
\item[(ii)] For $q=3$ and $d=3$,  we also have many other  instances of parameters for which our Theorem \ref{thm:3.2} shows that $r_3(n,3)<r^L_3(n,3)$. We tablet some of these parameters in the following table.
    \begin{center}
          Table II\\{Nonlinear codes with lower redundancy when $q=3$ and $d=3$} \medskip
     \begin{tabular}{|c|c|c|c|c|} \hline\hline
     $q$&$n$&$d$&Upper bound on $r_q(n,d)$&$r^L_q(n,d)$\\ \hline
       3&122&3&5.87&6\\ \hline
        3&123&3&5.87&6\\ \hline
       3&124&3&5.87&6\\ \hline
        3&125&3&5.87&6\\ \hline
         3&126&3&5.87&6\\ \hline
         3&127&3&5.87&6\\ \hline\hline

     \end{tabular}
     \end{center}

\end{itemize}
}\end{ex}

In Corollary~\ref{cor:3.4} and Corollary~\ref{cor:3.5}, we focus on minimum distance $d=3$ for defect $r=0$ and $r=1$ with code length $n$ belonging to some intervals.
\begin{cor}\label{cor:3.4} Let $q$ be a prime power. Then for any $n$ with $q+2\le n\le\frac{q^3}{4q-2}$, we have
\[r_q(n,3)<r^L_q(n,3).\]
\end{cor}
\begin{proof} First of all, by the fact that $N_q{(3,0)}=q+1$ given in Lemma \ref{lem:2.1}, we have $K_q(n,3)\le q^{n-3}$ for $n\ge q+2$.
In the interval $[n,2n)$, there must be a prime $\ell$. Thus, by Theorem \ref{thm:3.2}, we have
\begin{equation}~\label{eq:3.4.1}
A_q(n,3)\ge \frac{q^n}{(2q-1)\ell}> \frac{q^n}{(2q-1)\times 2n}\ge q^{n-3}\ge K_q(n,3).
\end{equation}
So, \[r_q(n,3)=n-\log_q A_q(n,3)< n-\log_q K_q(n,3)=r_q^L(n,3).\]
The proof is completed.
\end{proof}

\begin{cor}~\label{cor:3.5} Let $q$ be a prime power. Then for any $n$ with $q^2+q+1<n\le\frac{q^4}{4q-2}$, we have
\[r_q(n,3)<r_q^L(n,3).\]
\end{cor}
\begin{proof} By the fact that $N_q{(3,1)}=q^2+q+1$ given in Lemma \ref{lem:2.1}, we have $K_q(n,3)\le q^{n-4}$ for $n>q^2+q+1$.
In the interval $[n,2n)$, there must be a prime $\ell$. Thus, by Theorem \ref{thm:3.2}, we have
\begin{equation}~\label{eq:3.5.1}
A_q(n,3)\ge \frac{q^n}{(2q-1)\ell}> \frac{q^n}{(2q-1)\times 2n}\ge q^{n-4}\ge K_q(n,3).
\end{equation}
So, \[r_q(n,3)=n-\log_q A_q(n,3)< n-\log_q K_q(n,3)=r_q^L(n,3).\]
The proof is completed.
\end{proof}

\begin{ex}
{\rm Corollary~\ref{cor:3.4} and Corollary~\ref{cor:3.5} show that $r_q(n,3)<r_q^L(n, 3)$ if the code length $n$ belongs to some intervals.
The following table  lists these intervals for alphabet size $q=7,8$ and $9$.
 \begin{center}
 Table III\\{Nonlinear codes with lower redundancy when $d=3$} \medskip
     \begin{tabular}{|c|c|c|c|c|} \hline\hline
     $q$&Corollary~\ref{cor:3.4}&Corollary~\ref{cor:3.5}&Relation\\ \hline
       7&$9\le n\le 13$&$58\le n\le 92$&$r_7(n,3)<r_7^L(n,3)$\\ \hline
        8&$10\le n\le 17$&$74\le n\le 136$&$r_8(n,3)<r_8^L(n,3)$\\ \hline
       9&$11\le n\le 21$&$92\le n\le 197$&$r_9(n,3)<r_9^L(n,3)$\\ \hline\hline

     \end{tabular}
     \end{center}
}
\end{ex}

\begin{rmk}{\rm Note that the  upper bounds on length $n$ in~Corollary~\ref{cor:3.4} and Corollary ~\ref{cor:3.5} are not tight. The redundancy of codes is  $\log_q(2q-1)\ell$. However, we replace this prime $\ell$ by $2n$. This causes a larger redundancy.  If we choose the smallest prime $\ell$ satisfying $\ell\ge n$, we usually get smaller redundancy. We use the following numerical examples to illustrate this fact.
\begin{itemize}
\item[(1)] Fix $q=9$ and $d=3$. By Corollary~\ref{cor:3.4} we  have $r_9(n,3)<r_9^L(n,3)$ only for the range $n\in[11, 21].$ Now let us take $\ell=41$, then for any length $n$ satisfying $11\le n\le \ell=41,$ we have $r_9(n,3)=n-a_9(n,3)\le\log_9{(18-1)\times 41}=2.98<3=r_9^L(n,3)$, i.e., $r_9(n,3)<r_9^L(n,3)$ for any $n\in[11, 41].$
\item[(2)] Similarly, let $q=9$ and $d=3$.  By Corollary~\ref{cor:3.5} we only have $r_9(n,3)<r_9^L(n,3)$ only for the range $n\in[92, 92].$ Now let us take $\ell=383$, then for any length $n$ satisfying $92\le n\le \ell=383,$ we have $r_9(n,3)=n-a_9(n,3)\le\log_9{(18-1)\times 383}=3.99<4=r_9^L(n,3)$, i.e., $r_9(n,3)<r_9^L(n,3)$ for any $n\in[92, 383].$
\end{itemize}

}\end{rmk}

The above corollaries focus on nonlinear codes with minimum distance $d=3.$ The following result shows that for larger  minimum distance, we can also find a range of code length in which our codes perform better than linear codes.

\begin{cor}~\label{cor:3.7} Let $q$ be a prime power and let $n$ be a prime. If the Main MDS conjecture holds for $q$-ary MDS codes, then for any integer $n$ satisfying $q+2<n\le\left(\frac{q^d}{(d-1)(q-1)+1}\right)^{1/(d-2)}$ and any integer $3\le d\le n-2$ , we have
\[r_q(n,d)<r^L_q(n,d).\]
\end{cor}
\begin{proof} First of all, by the Main MDS conjecture, we have $K_q(n,d)\le q^{n-d}$.
Take $\ell=n$. Then, by Theorem \ref{thm:3.2}, we have
\begin{equation}\label{eq:z}A_q(n,d)\ge \frac{q^n}{((d-1)(q-1)+1)n^{d-2}}> q^{n-d}\ge K_q(n,d).\end{equation}
Hence, $r_q(n,d)= n-a_q(n,d)<d\le n- k_q(n,d)= r^L(n,d).$
The proof is completed.
\end{proof}

\begin{rmk}
{\rm In  Corollary~\ref{cor:3.7}, we assume that the code length $n$ is a prime. In fact, we do not have to make such an assumption. For any integer $n\ge 2$, we can simply replace $n$ by $2n$ in the denominator of \eqref{eq:z} to get a tighter upper bound on length $r_q(n,d)$.
}
\end{rmk}

\begin{ex}
{\rm Some examples are listed in the following from Corollary~\ref{cor:3.7}.
\begin{itemize}
\item[(1)] Fix $q=13$, $d=5$, then for $n\in\{17,19\}$,  we have $r_{13}(n,5)=4.96<5=r_{13}^L(n,5)$.
\item[(2)] Fix $q=17$, $d=5$, then $n\in\{19,23\}$,  we have $r_{17}(n,5)=4.79<5=r_{17}^L(n,5)$.
\item[(3)] Fix $q=32$, $d=7$, then $n\in\{37,41,43\}$,  we have $r_{32}(n,7)=6.72<7=r_{32}^L(n,7)$.
\end{itemize}
}
\end{ex}

\section{Decoding Algorithms}
In this section, we discuss decoding algorithm of the codes given in Section 3. Let us consider erasure error decoding algorithm and adversary error decoding algorithm separately. For erasure errors, the decoding algorithm mainly consists of solving an equation system. On the other hand, for adversary errors with larger error-correcting capability, the decoding algorithm is based on the decoding algorithm of  generalized Reed-Solomon codes.

Firstly, we introduce the erasure error decoding algorithm.
The code $\mC_d(\bb)$ has distance $d$. Hence, it can correct $d-1$ erasure errors. From our proof of Lemma \ref{lem:3.1}, we find that to decode $d-1$ erasure errors, it is sufficient to solve an equation system with $d-1$ variables over $\ZZ_\ell$. Thus the decoding complexity is $O(d^3\log^2\ell)=O(d^3\log^2n)$. In case $d$ is a constant, the decoding algorithm requires $O(\log^2n)$ bit operations. If $d$ is not a constant, we can use the following decoding algorithm for adversary errors to get a quasi-linear time $O(n\log^4n)$ for erasure errors.

To correct adversary errors, let $\tau=\left\lfloor\frac{d-1}2\right\rfloor$. We distinguish two cases: $\tau=1$ and  $\tau>1$. Let $\by=(y_1,y_2,\dots,y_n)\in[0,q-1]^n$ be a received word.

\begin{itemize}
\item[$\blacksquare$~\textbf{Case $1$}] {\textbf {$\tau=1$}}.
\begin{itemize}
\item[Step 1:] Compute $ t_0(\by)\pmod{(d-1)(q-1)+1} $
and  $t_j(\by)\pmod{\ell}$  for $1\le j\le d-2$. If  \[ t_0(\by)\equiv b_0 \pmod{(d-1)(q-1)+1},\quad t_j(\by)\equiv b_j\pmod{\ell}\]  for $1\le j\le d-2,$ output $\by$. Otherwise, go to Step 2.
\item[Step 2:] For every position $k\in[n]$, replace $y_k$ by an element $\Ga\in[0,q-1]\setminus\{y_k\}$ and compute  $ t_0(\bz)=t_0(\by)+(\Ga-y_k)\pmod{(d-1)(q-1)+1} $
and  $t_j(\bz)=k^j(\Ga-y_k)\pmod{\ell}$  for $1\le j\le d-2$, where $\bz$ is obtained by $\by$ by replaying $y_k$ with $\Ga$. If
\begin{equation}\label{eq:7}
t_0(\bz)\equiv b_0 \pmod{(d-1)(q-1)+1},\quad t_j(\bz)\equiv b_j\pmod{\ell}
  \end{equation}
  for $1\le j\le d-2$, output $\bz$. Otherwise, we try other elements of $[0,q-1]$ and position $k\in[n]$ until we find $k$ and $\Ga$ that satisfy \eqref{eq:7}.
\end{itemize}

Note that the complexity of the above algorithm in Case $1$ is $O(n\log^2 n)$ bit operations.

\item[$\blacksquare$~\textbf{Case $2$}] {\textbf {$\tau>1$.}}

Let $\mA$ be the dual code of the generalized Reed-Solomon code over $\ZZ_\ell$ with evaluation points $1,2,\dots, n$ and dimension $d-1$. Then $\mA$  is also a generalized Reed-Solomon code with parameters $[n,n-d+1,d]$. Assume that $\bGa$ is the codeword that was transmitted for the code $\mC_d(\bb)$. Then $\by-\bGa$ is the error vector when we transmit the codeword $\bGa$.
\begin{itemize}
\item[Step 1:]  Compute $ t_0(\by)\pmod{(d-1)(q-1)+1} $
and  $t_j(\by)\pmod{\ell}$  for $1\le j\le d-2$. View $ t_0(\by)\pmod{(d-1)(q-1)+1} $ as an element of $\ZZ_\ell$. Compute the vector $\bs:=\bb-(t_0(\by)\pmod{(d-1)(q-1)+1},t_1(\by)\pmod{\ell},\cdots, t_{d-2}(\by)\pmod{\ell})\in\ZZ_\ell^{d-1}$. Then $\bs$ is the syndrome of the received word $\by-\bGa$ for the  generalized Reed-Solomon code $\mA$.
\item[Step 2:] Find the error vector $\by-\bGa$ through the syndrome $\bs$ via a decoding of the  generalized Reed-Solomon code $\mA$.
\end{itemize}
The decoding complexity depends on the decoding algorithm of  generalized Reed-Solomon codes. The fastest decoding algorithm of Reed-Solomon codes has complexity $O(n\log (d-1)+(d-1)\log^2(d-1))$ bit operations \cite{{TH2002}}. Thus, it takes  $O((n\log (d-1)+(d-1)\log^2(d-1))\log^2n)=O(n\log^4n)$ bit operations. If $d$ is a constant, the the decoding complexity is in fact $O(n\log^2n)$ bit operations.
\end{itemize}

\end{document}